\documentclass[onecolumn, a4size, 12pt]{IEEEtran}
\usepackage{amsmath}
\usepackage{amssymb}
\usepackage{amsfonts}
\usepackage{graphicx}
\usepackage{epsfig}
\usepackage{subfigure}
\usepackage{psfrag}
\usepackage{xcolor}
\usepackage{amsfonts, bm}

\title{Joint Transmit Beamforming and Receive Power Splitting for MISO SWIPT Systems}
\author{Qingjiang Shi, Liang Liu, Weiqiang Xu, and Rui Zhang
\thanks{Q. Shi is with the School of Info. Sci. \& Tech., Zhejiang Sci-Tech University, Hangzhou 310018, China. He is also with the State Key Laboratory of Integrated Services Networks, Xidian University (e-mail:qing.j.shi@gmail.com).}
\thanks {L. Liu is with the
Department of Electrical and Computer Engineering, National
University of Singapore (e-mail:liu\_liang@nus.edu.sg).}
\thanks{W. Xu is with the School of Info. Sci. \& Tech. Zhejiang Sci-Tech University, Hangzhou 310018, China (e-mail:wqxu@zstu.edu.cn).}
\thanks{R.
Zhang is with the Department of Electrical and Computer Engineering,
National University of Singapore (e-mail:elezhang@nus.edu.sg). He is
also with the Institute for Infocomm Research, A*STAR, Singapore.}
}

\begin{document}
\maketitle %\thispagestyle{empty} \vspace{-0.3in}

\begin{abstract}
This paper studies a multi-user multiple-input single-output (MISO) downlink system for simultaneous
wireless information and power transfer (SWIPT), in which a set of single-antenna
mobile stations (MSs) receive information and energy simultaneously via power splitting (PS) from the
signal sent by a multi-antenna base station (BS). We aim to minimize the total transmission power at BS
by jointly designing transmit beamforming vectors and receive PS ratios for all MSs under their given signal-to-interference-plus-noise ratio (SINR) constraints for information decoding and harvested power constraints
for energy harvesting. First, we derive the sufficient and necessary condition for the feasibility
of our formulated problem. Next, we solve this
non-convex problem by applying the technique of semidefinite relaxation (SDR). We prove that SDR
is indeed tight for our problem and thus achieves its global optimum. Finally, we propose
two suboptimal solutions of lower complexity than the optimal solution based on the principle of separating the optimization of transmit
beamforming and receive PS, where the zero-forcing (ZF) and the SINR-optimal based transmit beamforming schemes are applied, respectively.
\end{abstract}


\setlength{\baselineskip}{1.5\baselineskip}

\newtheorem{definition}{\underline{Definition}}[section]
\newtheorem{fact}{Fact}
\newtheorem{assumption}{Assumption}
\newtheorem{theorem}{\underline{Theorem}}[section]
\newtheorem{lemma}{\underline{Lemma}}[section]
\newtheorem{corollary}{\underline{Corollary}}[section]
\newtheorem{proposition}{\underline{Proposition}}[section]
\newtheorem{example}{\underline{Example}}[section]
\newtheorem{remark}{\underline{Remark}}[section]
\newtheorem{algorithm}{\underline{Algorithm}}[section]
\newcommand{\mv}[1]{\mbox{\boldmath{$ #1 $}}}

\newcommand{\SINR}{\textrm{SINR}}
\newcommand{\trace}{{\rm Tr}}
\newcommand{\rank}{{\rm Rank}}
\newcommand{\diag}{{\rm diag}}
\newcommand{\st}{{\rm s.t.}}
\newcommand{\adj}{\rm{\textbf{adj}}}
\newcommand{\bI}{\mathbf{I}}
\newcommand{\bH}{\mathbf{H}}
\newcommand{\bP}{\mathbf{P}}
\newcommand{\bG}{\mathbf{G}}
\newcommand{\bT}{\mathbf{T}}
\newcommand{\bF}{\mathbf{F}}
\newcommand{\bQ}{\mathbf{Q}}
\newcommand{\hQ}{\hat{\bQ}}

\newcommand{\bA}{\mathbf{A}}
\newcommand{\bB}{\mathbf{B}}
\newcommand{\bC}{\mathbf{C}}
\newcommand{\bD}{\mathbf{D}}
\newcommand{\bE}{\mathbf{E}}
\newcommand{\bU}{\mathbf{U}}
\newcommand{\bV}{\mathbf{V}}

\newcommand{\bX}{\mathbf{X}}
\newcommand{\cX}{\mathcal{X}}
\newcommand{\bY}{\mathbf{Y}}
\newcommand{\cY}{\mathcal{Y}}
\newcommand{\bW}{\mathbf{W}}
\newcommand{\hY}{\hat{\bY}}
\newcommand{\barY}{\bar{\bY}}
\newcommand{\bS}{\mathbf{S}}

\newcommand{\bZ}{\mathbf{Z}}
\newcommand{\bOmega}{\mathbf{\Omega}}
\newcommand{\bLambda}{\mathbf{\Lambda}}
\newcommand{\bSigma}{\mathbf{\Sigma}}
\newcommand{\bPhi}{\mathbf{\Phi}}
\newcommand{\bTheta}{\mathbf{\Theta}}
\newcommand{\bK}{\mathbf{K}}

\newcommand{\bt}{\bm{t}}
\newcommand{\bx}{\bm{x}}
\newcommand{\barx}{\bar{\bm{x}}}
\newcommand{\barv}{\bar{\bm{v}}}
\newcommand{\by}{\bm{y}}
\newcommand{\bh}{\bm{h}}
\newcommand{\bp}{\bm{p}}
\newcommand{\hh}{\hat{\bm{h}}}
\newcommand{\hH}{\hat{\mathbf{H}}}
\newcommand{\hX}{\hat{\mathbf{X}}}

\newcommand{\bz}{\bm{z}}
\newcommand{\bb}{\bm{b}}
\newcommand{\bc}{\bm{c}}
\newcommand{\bu}{\bm{u}}
\newcommand{\bv}{\bm{v}}
\newcommand{\bn}{\bm{n}}
\newcommand{\br}{\bm{r}}
\newcommand{\bs}{\bm{s}}
\newcommand{\bw}{\bm{w}}
\newcommand{\Cdom}{\mathbb{C}}
\newcommand{\Rdom}{\mathbb{R}}
\newcommand{\expect}{\mathbb{E}}
\newcommand{\rgauss}{\mathcal{N}}
\newcommand{\cgauss}{\mathcal{CN}}
\newcommand{\mK}{\mathcal{K}}
\newcommand{\mL}{\mathbb{L}}

\section{Introduction}

Recently, simultaneous wireless information and power transfer (SWIPT) has drawn an upsurge of interests \cite{ZhangArXiv}-\cite{Clerckx13}. By SWIPT, mobile users are provided with both wireless data and energy accesses at the same time, which brings great convenience. However, there is one crucial issue for realizing SWIPT systems in practice, i.e., existing receiver circuits cannot decode information and harvest energy from the same received signal independently \textcolor{blue}{\cite{SunH2012}}. As a result, the receiver architecture design plays a significant role in determining the trade-offs between the end-to-end information versus energy transfer. Two practical receiver designs have been proposed for SWIPT, namely, time switching (TS) and power splitting (PS) \cite{ZhangArXiv}. With TS, the receiver switches over time between decoding information and harvesting energy, while with PS, the receiver splits the received signal into two streams of different power for decoding information and harvesting energy separately. Based on the PS scheme, a novel integrated receiver architecture for SWIPT was proposed in \cite{Rui12}, where the circuit for radio frequency (RF) to baseband conversion in the conventional information receiver is integrated to the front end of energy receiver via a rectifier, thus achieving a small form factor as well as energy saving. The TS and PS schemes have also been investigated for SWIPT over fading channels to exploit opportunistic information and energy transmissions \cite{Liu2013}, \cite{Rui13TCOM}. It is worth noting that theoretically, TS can be regarded as a special form of PS with only binary split power ratios, and thus in general PS achieves better rate-energy transmission trade-offs than TS \cite{ZhangArXiv}-\cite{Rui13TCOM}. However, in practice PS is implemented differently from TS since the former requires an RF signal splitter \cite{Wu} while the latter only needs a simpler switcher.

Another key concern for SWIPT is drastically decaying power transfer efficiency with the increasing transmission distance due to propogation pass loss. To tackle this problem, MIMO (multiple-input multiple-output) techniques by employing multiple antennas at the transmitter and/or receiver have been proposed in \cite{ZhangArXiv} to significantly improve the power transfer efficiency while still achieving high spectral efficiency for information transmission. Moreover, \cite{Xiang2012} extended \cite{ZhangArXiv} to the case with imperfect channel state information (CSI) at the transmitter. In \cite{RuiWCNC}, a MISO (multiple-input single-output) multicast SWIPT system with no CSI at the transmitter was studied, where random beamforming was proposed to improve the performance of opportunistic energy and information multicasting over quasi-static fading channels. Since in practice information and energy receivers have very different power sensitivity (e.g., $-60$dBm for the information receiver versus $-10$dBm for the energy receiver) \cite{ZhangArXiv}, \cite{Liu2013}, a ``near-far'' or receiver-location based scheduling for a MISO SWIPT system was proposed in \cite{Rui13}, where receivers that are close to the transmitter are scheduled for energy transmission, while others that are more distant away from the transmitter are scheduled for information transmission, to resolve the receiver \textcolor{blue}{sensitivity} issue. In \cite{RuiGlobecom}, the receiver-location based scheduling is extended to a MISO SWIPT system with the additional secrecy information transmission constraint. Furthermore, multi-antenna SWIPT systems have also been recently studied under the interference channel (IC) setup. A two-user MIMO IC was studied in \cite{Clerckx13} for SWIPT, where the achievable information and energy transmission trade-offs by different combinations of transmission modes at the two transmitters are characterized. The SWIPT system was also studied in the $K$-user MISO IC in \cite{Ottersten2013} and \cite{Shen12} based on PS and TS receivers, respectively.

\begin{figure}
\centering
\includegraphics[width=3.5in]{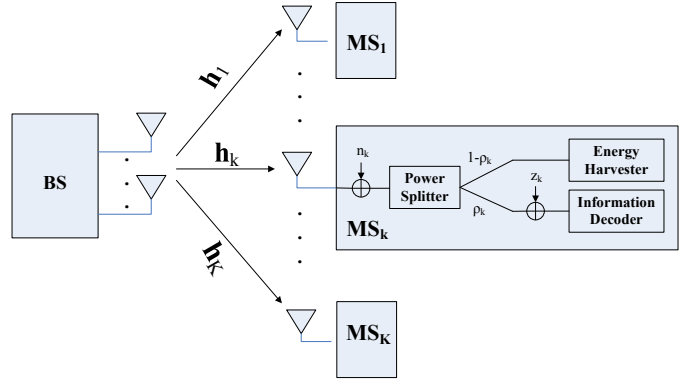}
\caption{A multi-user MISO SWIPT system, where each mobile station (MS) coordinates information decoding and energy harvesting via power splitting (PS).}
\label{fig:fig1}
\end{figure}

In this paper, we further study the multi-antenna and PS enabled SWIPT system by considering a MISO broadcast channel consisting of one multi-antenna base station (BS) and a set of $K\geq 1$ single-antenna mobile stations (MSs), as shown in Fig. \ref{fig:fig1}. We focus our study on PS receivers instead of TS receivers in \cite{Rui13}, \cite{RuiGlobecom}, such that each MS can receive both information and energy from the BS continuously at all time. Each MS is assumed to have its own required signal-to-interference-plus-noise ratio (SINR) for information decoding as well as harvested power amount for energy harvesting. Under the above two types of constraints at the same time, we study the joint design of transmit beamforming at BS and receive PS ratios at MSs to minimize the total transmission power. First, we derive the sufficient and necessary condition for the
feasibility of our formulated problem. Interestingly, it is shown that the feasibility of this problem only depends on the SINR constraints but not on the harvested power constraints. Next, we apply the technique of SDR \cite{Luo2010} to solve this non-convex problem, due to the coupled design variables of both beamforming vectors and PS ratios. We prove that SDR is indeed tight for our problem and thus it yields optimal beamforming  solution. Furthermore, we present two suboptimal designs of lower complexity, in which the transmit beamforming vectors are first designed based on the zero-forcing (ZF) and the SINR-optimal criteria, respectively, and then the receive PS ratios are optimized to minimize the transmission power. Finally, we compare the performance of our proposed optimal and suboptimal solutions by simulations.

The remainder of this paper is organized as follows.
Section II presents the MISO SWIPT system model and the formulation of our joint beamforming and PS design problem. Section III provides the feasibility condition for the formulated problem. Section IV presents the optimal
solution based on SDR and proves its optimality. Section V presents
two suboptimal solutions based on ZF and SINR-optimal beamforming, respectively. Section VI
provides numerical results for the performance comparison. Finally, Section VII concludes the paper.

\emph{Notations}: scalars are denoted by lower-case letters; bold-face lower-case letters are used for vectors,
while bold-face upper-case letters are for matrices. For a square matrix $\bA$, $\trace(\bA)$, $\rank(\bA)$, $\bA^T$ and $\bA^H$ denote its trace, rank, transpose and conjugate transpose, respectively, while $\bA\succeq0$ means that $\bA$ is a positive semidefinite matrix. $\bI_n$ denotes an $n$ by $n$ identity matrix. $||\cdot||$ denotes the Euclidean norm of a complex vector, while $|\cdot|$ denotes the absolute value of a complex scalar. The distribution of a circularly symmetric complex Gaussian (CSCG) random vector with mean $\bm{\mu}$ and covariance matrix $\bC$ is denoted by $\cgauss(\bm{\mu},\bC)$, and `$\sim$' stands for `distributed as'. Finally, $\Cdom^{m\times n}$ denotes the space of $m\times n$ complex matrices.

\section{System Model and Problem formulation}

This paper considers a multi-user MISO downlink SWIPT system consisting of one BS and $K$ MSs, denoted by ${\rm MS}_1,\cdots,{\rm MS}_K$, respectively, over a given frequency band, as shown in Fig. \ref{fig:fig1}. It is assumed that the BS is equipped with
$N_t>1$ antennas and each MS equipped with one antenna. We assume linear transmit precoding at BS, where each MS is assigned with one dedicated information beam. The complex baseband transmitted signal at BS is thus expressed as
\begin{equation}
\bx=\sum_{k=1}^K \bv_ks_k,
\end{equation}where $s_k$ denotes the transmitted data symbol for ${\rm MS}_k$, and $\bv_k$ is the corresponding transmit beamforming vector. It is assumed that $s_k$, $k=1,\cdots,K$, are independent and identically distributed (i.i.d.) CSCG random variables with zero mean and unit variance, denoted by $s_k\sim \mathcal{CN}(0,1)$.

%We assume the BS transmits a combination signal of all users' using a set of beamforming vectors $\bv_k\in\Cdom^{N_t\times 1}$. Each user $k$ receives signal and splits it into two parts. One part is used for further signal processing and symbol detection, and the other part is driven to the energy harvesting circuit for conversion to DC voltage and energy storage. The system model is depicted in Fig. \ref{fig:fig1} where $\rho_k$ is the power splitting ratio which means a fraction, $\rho_k$, of the received signal energy is used for information decoding while the remaining energy is for energy harvesting. We refer readers to \cite{ZhouArxiv} for more details about the separated type DPS-based receiver.

We assume the quasi-static flat-fading channel for all MSs and for convenience denote $\bh_k$ as the conjugated complex channel vector from BS to ${\rm MS}_k$. The received signal at ${\rm MS}_k$ is then given by
\begin{equation}
y_k = \bh_k^H \sum_{j=1}^K \bv_js_j+n_k, ~~ k=1,\cdots,K,
\end{equation}
where $n_k\sim \cgauss(0, \sigma_k^2)$ denotes the antenna noise at the receiver of ${\rm MS}_k$.

In this paper, we assume each MS applies PS to coordinate the processes of information decoding and energy harvesting from the received signal \cite{ZhangArXiv}. Specifically, as shown in Fig. \ref{fig:fig1}, the received signal at each MS is split to the information decoder (ID) and the energy harvester (EH) by a power spitter, which divides an $\rho_k$ ($0\leq \rho_k \leq 1$) portion of the signal power to the ID, and the remaining $1-\rho_k$ portion of power to the EH. As a result, the signal split to the ID of ${\rm MS}_k$ is expressed as
\begin{equation}
y_k^{{\rm ID}} = \sqrt{\rho_k}\left(\bh_k^H\sum_{j=1}^K\bv_js_j+n_k\right)+z_k, ~~ k=1,\cdots,K,
\end{equation}
where $z_k\sim \cgauss(0, \delta_k^2)$ is the additional noise introduced by the ID at ${\rm MS}_k$. Accordingly, the SINR at the ID of ${\rm MS}_k$ is given by\begin{equation}
\SINR_k=\frac{\rho_k|\bh_k^H\bv_k|^2}{\rho_k\sum_{j\neq k}|\bh_k^H\bm{v}_j|^2+\rho_k\sigma_k^2+\delta_k^2}, ~~ k=1,\cdots,K.
\end{equation}

On the other hand, the signal split to the EH of ${\rm MS}_k$ is expressed as
\begin{equation}
y_k^{{\rm EH}} = \sqrt{1-\rho_k}\left(\bh_k^H\sum_{j=1}^K \bv_js_j+n_k\right), ~~ k=1,\cdots,K.
\end{equation}Then, the harvested power by the EH of ${\rm MS}_k$ is given by
\begin{equation}
E_k = \zeta_k (1-\rho_k)\left(\sum_{j=1}^K |\bh_k^H\bv_j|^2+\sigma_k^2\right), ~~ k=1,\cdots,K,
\end{equation}
where $\zeta_k\in(0~1]$ denotes the energy conversion efficiency at the EH of ${\rm MS}_k$.

In order to realize a continuous information transfer, each ${\rm MS}_k$ requires its SINR to be above a given target, denoted by $\gamma_k$, at all time. In the meanwhile, each ${\rm MS}_k$ also requires that its harvested power needs to be no smaller than a given threshold, denoted by $e_k$, to maintain its receiver operation. Under the above two types of constraints, we aim to minimize the total transmission power at BS by jointly designing transmit beamforming vectors, $\{\bv_k\}$, and receive PS ratios, $\{\rho_k\}$, at all MSs, i.e.,

%Different from a general wireless communication system, there are two types of QoS requirement in an SWIPT system. To guarantee successful symbol detection, it is required for each user $k$ that the achieved SINR should be higher than a threshold $\gamma_k$, corresponding to SINR constraint. Also, to ensure substantiable power supply for user $k$, the harvested power $E_k$ should be higher than a threshold $\hat{e}_k$, corresponding to EH constraint. Under these two kinds of QoS constraints, we seek to minimize the transmission power by joint beamforming and power splitting. The corresponding optimization problem, called joint beamforming and power splitting (JBPS) problem, can be mathematically written as
\begin{equation}
\label{eq:P1}
\begin{split}
\min_{\{\bv_k, \rho_k\}}  & ~~ \sum_{k=1}^K ||\bv_k||^2\\
\st & ~~ \frac{\rho_k|\bh_k^H\bv_k|^2}{\sum_{j\neq k}\rho_k|\bm{h}_k^H\bv_j|^2+\rho_k\sigma_k^2+\delta_k^2}{\geq} \gamma_k, ~~ \forall k,\\
&~~\zeta_k(1-\rho_k)\left(\sum_{j=1}^K |\bh_k^H\bv_j|^2+\sigma_k^2\right)\geq e_k, ~~ \forall k, \\
&~~0< \rho_k< 1, ~~ \forall k.
\end{split}
\end{equation}Notice that in this paper we consider the general case that all MSs have non-zero SINR and harvested power targets, i.e., $\gamma_k>0$ and $e_k>0$, $\forall k$; thus, the receive PS ratios at all MSs should satisfy $0<\rho_k<1$, $\forall k$, as given by the last constraint in \eqref{eq:P1}.

For convenience, problem \eqref{eq:P1} is referred to as the joint beamforming and power splitting (JBPS) problem in the sequel. Note that JBPS is non-convex due to \textcolor{blue}{not only the coupled beamforming vectors $\{\bv_k\}$ and PS ratios $\{\rho_k\}$ in both the SINR and harvested power constraints but also all the quadratic terms involving $\{\bv_k\}$}. It is also worth noting that if we fix $\rho_k$'s with $0<\rho_k<1$, $\forall k$, the resulting beamforming optimization problem over $\{\bv_k\}$ is still non-convex due to the harvested power constraints with $e_k>0$, $\forall k$. Finally, notice that if we remove all the harvested power constraints and let $\rho_k\rightarrow 1$, $\forall k$, the above problem reduces to the conventional power minimization problem subject to only SINR constraints in the MISO broadcast channel, which can be efficiently solved by existing methods \cite{Ottersten2001,Codreanu07,Boche2004}. In the following, we first derive the sufficient and necessary condition for the feasibility of the JBPS problem in \eqref{eq:P1}, and then propose both optimal and suboptimal solutions to this problem. \textcolor{blue}{Note that, for practical implementation of all solutions, the computation takes place at the BS and then the BS sends each $\rho_k$ to the corresponding MS}.

%Clearly, problem \eqref{eq:P1} is nonconvex and difficult to solve. Particularly, even when $\rho_k$'s are fixed, the problem is still more difficult than the general SINR-constrained power minimization problem (which can be cast as a second order cone programming (SOCP) by phase rotation and thus solved efficiently\cite{Ottersten2001}) due to the extra nonconvex EH constraints. In \cite{Ottersten2013}, researchers have considered a similar problem for interference channel but only propose suboptimal solutions. Here, we attempt to provide an optimal solution to problem \eqref{eq:P1}. Moreover, we show the asymptotic optimality of some suboptimal methods.

\section{When is the JBPS Problem Feasible?}\label{sec2}
Before we proceed to solve the JBPS problem in \eqref{eq:P1}, in this section we study its feasibility condition for a given set of $\gamma_k>0$ and $e_k>0$, $\forall k$. First, we have the following lemma.

\begin{lemma}\label{lem:no EH}
Problem \eqref{eq:P1} is feasible if and only if the following problem is feasible.
\begin{equation}
\label{eq:no EH}
\begin{split}
\textrm{find} &~~~~\{\bv_k, \rho_k\} \\
\st & ~~~ \frac{\rho_k|\bh_k^H\bv_k|^2}{\sum_{j\neq k}\rho_k|\bm{h}_k^H\bv_j|^2+\rho_k\sigma_k^2+\delta_k^2}{\geq} \gamma_k, ~~ \forall k, \\
& ~~~ 0< \rho_k< 1, ~~ \forall k.
\end{split}
\end{equation}
\end{lemma}

\begin{proof}
Please refer to Appendix \ref{appendix1}.
\end{proof}

Lemma \ref{lem:no EH} indicates that the feasibility of problem \eqref{eq:P1} does not depend on its harvested power constraints. The following lemma further simplifies the feasibility test for problem \eqref{eq:no EH}.

\begin{lemma}\label{lem:no rho}
Problem \eqref{eq:no EH} is feasible if and only if the following problem is feasible.
\begin{equation}
\label{eq:no rho}
\begin{split}
\textrm{find} &~~~~\{\bv_k\} \\
\st & ~~~ \frac{|\bh_k^H\bv_k|^2}{\sum_{j\neq k}|\bm{h}_k^H\bv_j|^2+\sigma_k^2+\delta_k^2}{\geq} \gamma_k, ~~ \forall k.
\end{split}
\end{equation}
\end{lemma}

\begin{proof}
Please refer to Appendix \ref{appendix2}.
\end{proof}

Lemma \ref{lem:no rho} indicates that the feasibility of problem \eqref{eq:no EH} can be checked by letting $\rho_k\rightarrow 1$, $\forall k$. Combining Lemmas \ref{lem:no EH} and \ref{lem:no rho}, it follows that the feasibility condition for problem \eqref{eq:P1} must be the same as that of problem \eqref{eq:no rho}. Problem \eqref{eq:no rho} is a well-known SINR feasibility problem and its feasibility region over $\gamma_k$'s has been characterized in the literature \textcolor{blue}{\cite{Hunger10}}. We thus have the following proposition, which presents the sufficient and necessary condition for the feasibility of problem \eqref{eq:no rho}.

%Note that problem \eqref{eq:no rho} can be equivalently transformed into a SOCP \cite{Ottersten2001} and thus efficiently solved %by existing software, e.g., CVX \cite{}. However, in the following, we present a closed-form sufficient and necessary condition to %test the feasibility of the JBPS problem.

\begin{proposition}\label{pro:condition}\cite[Theorem III.1]{Hunger10}
Problem \eqref{eq:no rho} is feasible if and only if the SINR targets $\gamma_k$'s satisfy the following condition:
\begin{equation}\label{eq:condition}
\sum_{k=1}^K\frac{\gamma_k}{1+\gamma_k}\leq\rank(\bH).
\end{equation}where $\bH\triangleq[\bh_1~\bh_2~\ldots~\bh_K]$.
\end{proposition}

%\begin{proof}
%Please refer to Appendix \ref{appendix3}.
%\end{proof}

Therefore, the feasibility of the JBPS problem in \eqref{eq:P1} for a given set of $\gamma_k$'s and $e_k$'s can be simply verified by checking whether $\gamma_k$'s satisfy the condition given in Proposition \ref{pro:condition}. Without loss of generality, in the rest of this paper we assume that problem \eqref{eq:P1} is feasible, unless stated otherwise.
\section{Optimal Solution}\label{sec3}

In this section, we apply the celebrated technique of SDR to solve the JBPS problem in \eqref{eq:P1} optimally. Define $\bX_k=\bv_k\bv_k^H$, $\forall k$. It then follows that $\rank(\bX_k)\leq 1$, $\forall k$. By ignoring the above rank-one constraint for all $\bX_k$'s, the SDR of problem \eqref{eq:P1} is given by

\begin{equation}
\label{eq:P2}
\begin{split}
\min_{\{\bX_k, \rho_k\}} & ~~ \sum_{k=1}^K\trace{(\bX_k)}\\
\st & ~~ \frac{\rho_k\bh_k^H\bX_k\bh_k}{\sum_{j\neq k}\rho_k\bh_k^H\bX_j\bh_k+\rho_k\sigma_k^2+\delta_k^2}{\geq} \gamma_k, ~~ \forall k, \\
&~~ \zeta_k(1-\rho_k)\left(\sum_{j=1}^K \bh_k^H\bX_j\bh_k+\sigma_k^2\right)\geq e_k, ~~ \forall k, \\
&~~0< \rho_k< 1, ~~ \forall k, \\ & ~~ \bX_k\succeq \mathbf{0}, ~~ \forall k.
\end{split}
\end{equation}

Problem \eqref{eq:P2} is still non-convex in its current form since both the SINR and harvested power constraints involve coupled $\bX_k$'s and $\rho_k$'s. However, problem \eqref{eq:P2} can be reformulated as the following problem.
\begin{equation}
\label{eq:P3}
\begin{split}
\min_{\{\bX_k, \rho_k\}} & ~~ \sum_{k=1}^K \trace{(\bX_k)}\\
\st & ~~ \frac{1}{\gamma_k}\bh_k^H\bX_k\bh_K-\sum_{j\neq k}\bh_k^H\bX_j\bh_k\geq \sigma_k^2+\frac{\delta_k^2}{\rho_k}, ~~ \forall k,\\
& ~~ \sum_{j=1}^K \bh_k^H\bX_j\bh_k\geq \frac{e_k}{\zeta_k(1-\rho_k)}-\sigma_k^2, ~~ \forall k, \\
&~~\bX_k\succeq 0, ~~ \forall k,\\
&~~0<\rho_k< 1, ~~ \forall k.
\end{split}
\end{equation}
Note that problem \eqref{eq:P3} is convex due to the fact that both $\frac{1}{\rho_k}$ and $\frac{1}{1-\rho_k}$ are convex functions over $\rho_k$ with $0<\rho_k<1$. Let $\{\bX_k^\ast\}$ and $\{\rho_k^\ast\}$ denote the optimal solution to problem \eqref{eq:P3}. If $\{\bX_k^\ast\}$ satisfies $\rank(\bX_k^\ast)=1$, $\forall k$, then the optimal beamforming solution $\bv_k^\ast$ to problem \eqref{eq:P1} can be obtained from the eigenvalue decomposition (EVD) of $\bX_k^\ast$, $k=1,\cdots,K$, and the optimal PS solution of problem \eqref{eq:P1} is also given by the associated $\rho_k^\ast$'s. Otherwise, if there exists any $k$ such that $\rank(\bX_k^\ast)>1$, then in general the solution $\{\bX_k^\ast\}$ and $\{\rho_k^\ast\}$ of problem \eqref{eq:P3} is not necessarily optimal for problem \eqref{eq:P1}. In the following proposition, we show that it is indeed true that for problem \eqref{eq:P3}, the solution satisfies $\rank(\bX_k^\ast)=1$, $\forall k$, i.e., the SDR is tight.

\begin{proposition}\label{pro:rank one}

For problem \eqref{eq:P3} given $\gamma_k>0$ and $e_k>0$, $\forall k$, we have

\begin{itemize}
\item[1)] $\{\bX_k^\ast\}$ and $\{\rho_k^\ast\}$ satisfy the first two sets of constraints of problem \eqref{eq:P3} with equality;
\item[2)] $\{\bX_k^\ast\}$ satisfies $\rank(\bX_k^\ast)=1$, $\forall k$.
\end{itemize}
\end{proposition}

\begin{proof}
Please refer to Appendix \ref{appendix4}.
\end{proof}

The second part of Proposition \ref{pro:rank one} indicates that the rank relaxation on $\bX_k$'s in problem \eqref{eq:P2} or \eqref{eq:P3} results in no loss of optimality to problem \eqref{eq:P1}; thus, the optimal solution to problem \eqref{eq:P1} can be obtained via solving problem \eqref{eq:P3} by interior-point algorithm\cite{cvx_book} using existing software, e.g., CVX \cite{cvx2012}. \textcolor{blue}{According to \cite[sec. 6.6.3]{Ben2001}, it is known that\footnote{For better efficiency, the dual problem of problem \eqref{eq:P3} could be solved instead of problem \eqref{eq:P3}. The complexity of solving the dual problem is $O\left(\sqrt{KN_t}\left(K^3N_t^2+K^2N_t^3\right)\right)$ \cite[sec. 6.6.3]{Ben2001}.} the complexity of the interior-point algorithm for solving problem \eqref{eq:P3} is $O\left(\sqrt{KN_t}\left(K^3N_t^2+K^2N_t^3\right)\right)$}. To summarize, one algorithm for solving problem \eqref{eq:P1} is given in Table \ref{tab:beamforming_alg} as Algorithm 1. Furthermore, the first part of Proposition \ref{pro:rank one} suggests that with the optimal beamforming and PS solution, both the SINR and harvested power constraints in problem \eqref{eq:P1} should hold with equality for all MSs.

\begin{table}[h]
\centering
\caption{Algorithm 1: Optimal algorithm for problem \eqref{eq:P1}}\label{tab:beamforming_alg}
\begin{tabular}{|p{3.in}|}
\hline
\begin{itemize}
\item [1.] \; Check whether the SINR targets $\gamma_k$, $k=1,\cdots,K$, satisfy the feasibility condition given in \eqref{eq:condition}. If no, exit the algorithm; otherwise, go to step 2.
\item[2.] \;  Solve problem \eqref{eq:P3} by CVX, and obtain the optimal solution as $\{\bX_k^\ast\}$ and $\{\rho_k^\ast\}$.
\item[3.] \;  Obtain $\bv_k^*$ by EVD of $\bX_k^\ast$, $k=1,\cdots,K$.
%its (i,j)-th entry being
%$$\bA_{ij}=\left\{\begin{split}&\frac{1}{\gamma_i}|\bh_i\bv_i|^2,~~i=j\\&-|\bh_i\bv_j|^2,~~i\neq j\end{split}\right.$$

\end{itemize}
\\
\hline
\end{tabular}
\end{table}

%\begin{remark}
%Note that in \cite{Ottersten2013} a similar JBSP to problem \eqref{eq:P1} has been studied, but in the interference channel. By fixing the beamformers and optimizing the power allocation and power splitting, the authors proposed two sub-optimal algorithms to minimize the total transmission power subject to the SINR and EH constraints. However, in this paper, we obtain the optimal solution to the JBSP problem based on the SDR technique in the broadcast channel. Note that, due to the different channel setup in interference channel, the key proof for the rank one property of the SDR solution of the JBPS problem \eqref{eq:P1} cannot be extended to the interference channel case.
%\end{remark}

\section{Suboptimal Solutions}

The optimal solution to problem \eqref{eq:P1} derived in the previous section requires a joint optimization of the beamforming vectors and PS ratios. In this section, we present two suboptimal algorithms for problem \eqref{eq:P1} to achieve lower complexity. Both algorithms are based on the approach of separately designing the beamforming vectors and PS ratios, while the ZF-based and SINR-optimal based criteria are applied for the beamforming design in the two algorithms, respectively. \textcolor{blue}{We also study the asymptotic optimality of both suboptimal algorithms.}

\subsection{ZF Beamforming}
When $N_t\geq K$, the ZF beamforming scheme can be used to eliminate the multiuser interference by restricting $\bv_k$'s to satisfy $\bh_i^H\bv_k=0$, $\forall i\neq k$, which simplifies the beamforming design. With ZF transmit beamforming, problem \eqref{eq:P1} reduces to the following problem.
\begin{equation}
\label{eq:P1ZF}
\begin{split}
\min_{\{\bv_k, \rho_k\}} & ~~ \sum_{k=1}^K ||\bv_k||^2\\
\st & ~~ \frac{\rho_k|\bh_k^H\bv_k|^2}{\rho_k\sigma_k^2+\delta_k^2}{\geq} \gamma_k, ~~ \forall k, \\
&~~\zeta_k(1-\rho_k)(|\bh_k^H\bv_k|^2+\sigma_k^2)\geq e_k, ~~ \forall k, \\
&~~\bH_k^H\bv_k=0, ~~ \forall k, \\
&~~0<\rho_k< 1, ~~ \forall k.
\end{split}
\end{equation}
where $\bH_k\triangleq [\bh_1~\cdots~\bh_{k-1} ~\bh_{k+1}~\cdots~\bh_K]\in\Cdom^{N_t\times (K-1)}$. \textcolor{blue}{It is readily seen that problem \eqref{eq:P1ZF} must be feasible if $N_t\geq K$ and furthermore $\bh_k$'s are not linear dependent.} The following proposition then gives the optimal solution to problem \eqref{eq:P1ZF} in closed-form.

\begin{proposition}\label{propZF}
Let $\bU_k$ denote the orthogonal basis of the null space of $\bH_k^H$, $k=1,\cdots,K$. Define $\alpha_k\triangleq\frac{e_k}{\zeta_k(\gamma_k+1)\sigma_k^2}$ and $\beta_k\triangleq\frac{\gamma_k\delta_k^2}{(\gamma_k+1)\sigma_k^2}$, $\forall k$. The optimal solution to problem \eqref{eq:P1ZF} is given by
\begin{align}
&\bar{\rho}_k^* =\frac{-(\alpha_k+\beta_k-1)+\sqrt{(\alpha_k+\beta_k-1)^2+4\beta_k}}{2},  ~~ k=1,\cdots,K, \label{eq:ZFsol1}\\
&\bar{\bv}_k^*=\sqrt{\gamma_k\left(\sigma_k^2+\frac{\delta_k^2}{\bar{\rho}_k^*}\right)}\frac{\bU_k\bU_k^H\bh_k}{\bh_k^H\bU_k\bU_k^H\bh_k}, ~~ k=1,\cdots,K. \label{eq:ZFsol2}
\end{align}
\end{proposition}
\begin{proof}
Please refer to Appendix \ref{appendix5}.
\end{proof}

In Table \ref{tab:ZF_alg}, we summarize the above algorithm for problem \eqref{eq:P1} based on the ZF transmit beamforming as Algorithm 2. \textcolor{blue}{Clearly, the complexity of Algorithm 2 is dominated by the $K$ times of SVD operations. Since each SVD operation takes an $O((K-1)^3+N_t^2(K-1))$ complexity, the complexity of Algorithm 2 is $O(K^4+K^2N_t^2)$.}

\begin{table}
\centering
\caption{Algorithm 2: ZF beamforming based suboptimal algorithm for problem \eqref{eq:P1}}\label{tab:ZF_alg}
\begin{tabular}{|p{3.in}|}
\hline
\begin{itemize}
\item [1.]\; Set $\bH_k=\left[\bh_1,\ldots,\bh_{k-1}, \bh_{k+1},\ldots, \bh_K\right]$, $\forall k$.
\item [2.] \; Set $\bU_k=\textrm{null}(\bH_k^H)$, $\forall k$, where `null($\cdot$)' is a Matlab function which computes the orthonormal basis for the null space of a matrix using singular value decomposition (SVD).
\item[3.]\; Set $\alpha_k=\frac{e_k}{\zeta_k(\gamma_k+1)\sigma_k^2}$ and $\beta_k=\frac{\gamma_k\delta_k^2}{(\gamma_k+1)\sigma_k^2}$, $\forall k$.
\item[4.] \;  Set $\bar{\rho}_k^* =\frac{-(\alpha_k+\beta_k-1)+\sqrt{(\alpha_k+\beta_k-1)^2+4\beta_k}}{2}$, $\forall k$.
\item [5.] \; Set $\bar{\bv}_k^*=\sqrt{\gamma_k\left(\sigma_k^2+\frac{\delta_k^2}{\bar{\rho}_k^*}\right)}\frac{\bU_k\bU_k^H\bh_k}{\bh_k^H\bU_k\bU_k^H\bh_k}$, $\forall k$.
\end{itemize}
\\
\hline
\end{tabular}
\end{table}

\begin{remark}\label{remark2}
In Proposition \ref{propZF}, we have assumed that $\sigma_k^2>0$, $\forall k$, i.e., the antenna noise at each MS is non-zero. In practice, the antenna noise power $\sigma_k^2$ is much smaller than the ID noise power $\delta_k^2$ \cite{Rui12}. Thus, if we neglect the antenna noise power by setting $\sigma_k^2=0$, $\forall k$, the optimal PS solution can be shown (see \eqref{eq:rho_equation} in Appendix \ref{appendix5}) to be \begin{align}\label{eq:r}\bar{\rho}_k^*=\frac{\gamma_k\delta_k^2}{\frac{e_k}{\zeta_k}+\gamma_k\delta_k^2},\end{align}which has an simpler form than that given in \eqref{eq:ZFsol1} for $\sigma_k^2>0$.
\end{remark}

\subsection{SINR-Optimal Beamfoming}

Algorithm 2 works only when $N_t\geq K$ due to the ZF transmit beamforming. Alternatively, we can apply SINR-optimal transmit beamforming, which works for arbitrary values of $N_t$ and $K$. The SINR-optimal transmit beamforming vectors can be first obtained by solving the following power minimization problem with only SINR constraints, i.e.,
\begin{equation}
\label{eq:power_min_SINR_only}
\begin{split}
\min_{\{\bv_k\}} & ~~~ \sum_{k=1}^K||\bv_k||^2\\
\st & ~~~ \frac{|\bh_k^H\bv_k|^2}{\sum_{j\neq k}|\bm{h}_k^H\bv_j|^2+\sigma_k^2+\delta_k^2}{\geq} \gamma_k, ~~ \forall k.
\end{split}
\end{equation}
Note that problem \eqref{eq:power_min_SINR_only} is feasible if and only if problem \eqref{eq:P1} is feasible (see Section \ref{sec2}). Problem \eqref{eq:power_min_SINR_only} has been well studied in the literature and can be efficiently solved by existing techniques \cite{Ottersten2001,Codreanu07,Boche2004}. Let $\{\hat{\bv}_k\}$ denote the solution of problem \eqref{eq:power_min_SINR_only}. Next, we scale up the beamformers $\{\hat{\bv}_k\}$ by a common factor $\sqrt{\alpha}$ and then jointly optimize $\alpha$ and receive PS ratios $\rho_k$'s to satisfy both the SINR and harvested power constraints in problem \eqref{eq:P1} and yet minimize the transmission power. Specifically, we consider the following problem with given $\{\hat{\bv}_k\}$.
\begin{equation}
\label{eq:P1-subopt2}
\begin{split}
\min_{\alpha, \{\rho_k\}}  & ~~ \sum_{k=1}^K \alpha||\hat{\bv}_k||^2\\
\st & ~~ \frac{\rho_k\alpha|\bh_k^H\hat{\bv}_k|^2}{\sum_{j\neq k}\rho_k\alpha|\bm{h}_k^H\hat{\bv}_j|^2+\rho_k\sigma_k^2+\delta_k^2}{\geq} \gamma_k, ~~ \forall k,\\
&~~\zeta_k(1-\rho_k)(\alpha\sum_{j=1}^K |\bh_k^H\hat{\bv}_j|^2+\sigma_k^2)\geq e_k, ~~ \forall k, \\
&~~0< \rho_k< 1, ~~ \forall k,\\
&~~\alpha>1.
\end{split}
\end{equation}
Since all the SINR constraints in problem \eqref{eq:power_min_SINR_only} can be shown to hold with equality by $\{\hat{\bv}_k\}$ \cite{Boche2004}, the SINR constraints in problem \eqref{eq:P1-subopt2} hold with equality when $\alpha=1$ and $\rho_k=1$, $\forall k$. However, to satisfy the additional harvested power constraints in problem \eqref{eq:P1-subopt2}, it is required that $0<\rho_k<1$, $\forall k$. As a result, we need $\alpha>1$ to satisfy both the SINR and harvested power constraints in problem \eqref{eq:P1-subopt2}, which is given as the last constraint in \eqref{eq:P1-subopt2}. Next, we present the following proposition for problem \eqref{eq:P1-subopt2}.

\begin{proposition}\label{prop:subopt_alg2}
\
\begin{itemize}
\item[1)] Problem \eqref{eq:P1-subopt2} is feasible if and only if problem \eqref{eq:P1} is feasible;
\item[2)] Define $c_k\triangleq\frac{|\bh_k^H\hat{\bv}_k|^2}{\gamma_k}-\sum_{j\neq k}|\bm{h}_k^H\hat{\bv}_j|^2$ and $d_k\triangleq\sum_{j=1}^K|\bm{h}_k^H\hat{\bv}_j|^2$, $k=1,\cdots,K$. For each $k$, let $\bar{\alpha}_k$ be the largest real root of the following quadratic equation:
$$\frac{\delta_k^2}{\alpha c_k-\sigma_k^2}+\frac{e_k}{\zeta_k(\alpha d_k+\sigma_k^2)}= 1.$$ Then $\tilde{\alpha}^*=\max_{1\leq k\leq K}\bar{\alpha}_k$ and $\tilde{\rho}_k^*=\frac{\delta_k^2}{\tilde{\alpha}^* c_k-\sigma_k^2}$, $\forall k$, is the optimal solution to problem \eqref{eq:P1-subopt2}.
\end{itemize}
\end{proposition}
\begin{proof}
Please refer to Appendix \ref{sec:subopt_alg2}.
\end{proof}

\begin{table}
\centering
\caption{Algorithm 3: SINR-optimal beamforming based suboptimal algorithm for problem \eqref{eq:P1}}\label{tab:FBPS_alg}
\begin{tabular}{|p{3.in}|}
\hline
\begin{itemize}

\item [1.]\; Obtain $\hat{\bv}_k$'s by solving problem \eqref{eq:power_min_SINR_only}.
\item [2.] \; Set $c_k=\frac{|\bh_k^H\hat{\bv}_k|^2}{\gamma_k}-\sum_{j\neq k}|\bm{h}_k^H\hat{\bv}_j|^2, \forall k$.
\item[3.]\; Set $d_k=\sum_{j=1}^K|\bm{h}_k^H\hat{\bv}_j|^2$, $\forall k$.
\item[4.] \;  Set $\bar{\alpha}_k$ as the largest real root of the following equation:
$$\frac{\delta_k^2}{\alpha c_k-\sigma_k^2}+\frac{e_k}{\zeta_k(\alpha d_k+\sigma_k^2)}= 1, \forall k.$$
\item [5.] \; Set $\tilde{\alpha}^\ast=\max_{1\leq k\leq K}\bar{\alpha}_k$.
\item [6.]\;  Set $\tilde{\rho}_k^*=\frac{\delta_k^2}{\tilde{\alpha}^\ast c_k-\sigma_k^2}$ and $\tilde{\bv}_k^*=\sqrt{\tilde{\alpha}^\ast}\hat{\bv}_k$, $\forall k$.
\end{itemize}
\\
\hline
\end{tabular}
\end{table}

With $\tilde{\alpha}^\ast$, the corresponding beamforming vectors can be obtained as $\tilde{\bv}_k^*=\sqrt{\tilde{\alpha}^*}\hat{\bv}_k$, $\forall k$. In Table \ref{tab:FBPS_alg}, the suboptimal algorithm  for problem \eqref{eq:P1} based on the SINR-optimal transmit beamforming is summarized as Algorithm 3. \textcolor{blue}{Clearly, the complexity of Algorithm 3 is dominated by solving problem \eqref{eq:power_min_SINR_only}, which requires\footnote{Problem \eqref{eq:power_min_SINR_only} can be iteratively solved using the existing uplink-downlink duality \cite{Boche2004}. The complexity of solving problem \eqref{eq:power_min_SINR_only} is dominated by the inversion operations of $K$ $N_t$--by--$N_t$ matrices (corresponding to beamforming directions update) and one $(K{+}1)$--by--$(K{+}1)$ matrix (corresponding to power allocation update). Therefore, the complexity of Algorithm 3 is $O(K^3+KN_t^3)$. Note that we have neglected the number of iterations in the above calculation since it is observed from simulations that the algorithm has very fast convergence (the number of iterations is $10\sim20$ in general).} $O(K^3+KN_t^3)$.
\begin{remark}
To summarize, the complexity of the optimal solution is $O\left(\sqrt{KN_t}\left(K^3N_t^2+K^2N_t^3\right)\right)$, while the complexity of the ZF-based suboptimal solution and the SINR-optimal suboptimal solution are $O(K^4+K^2N_t^2)$ and $O(K^3+KN_t^3)$, respectively. Note that the ZF-based suboptimal solution works only when $N_t\geq K$. Hence, we can see that, in terms of computational complexity, the suboptimal solutions are better than the optimal solution and the ZF-based suboptimal solution is better than the SINR-optimal beamforming based suboptimal solution.
\end{remark}}

\subsection{\textcolor{blue}{Asymptotic Optimality}}
Algorithms 2 and 3 are suboptimal in general. However, we will show that the two suboptimal solutions can achieve optimality in terms of minimum transmit power when the SINR target $\gamma_k$'s are asymptotically high. The intuition behind the asymptotic optimality is that the power allocation for the optimal solution and suboptimal solutions all become extremely large as $\gamma_k$'s go to infinity and furthermore the interference terms (in the SINR expression) vanish in order to keep the SINR constraints feasible with increasing $\gamma_k$'s. We summarize this result in the following proposition.
\begin{proposition}\label{prop:asym_opt}
The two suboptimal solutions are both asymptotically optimal when $\gamma_k$'s go to infinity.
\end{proposition}
\begin{proof}
Please refer to Appendix F.%\ref{sec:proof_asym_opt}.
\end{proof}
%\begin{remark}
%It can be shown that, when the number of BS antennas $N_t$ is no less than the number of users $K$ and all the SINR targets $\gamma_k$'s approach infinity, the fixed beamformers $\hat{\bv}_k$'s obtained from solving problem \eqref{eq:power_min_SINR_only} need to satisfy the zero-forcing condition. Moreover, the scaling factor $\alpha$ (implying the transmission power of all users) will go to infinity as $\gamma_k\rightarrow\infty$ for all $k$. Hence, the proposed fixed beamforming and power splitting scheme can also achieve asymptotical optimality. This will be verified by simulations later.
%\end{remark}

\section{Simulation results}
In this section, we numerically evaluate the performance of the proposed beamforming and power splitting algorithms in MISO SWIPT systems. We assume there are $K=4$ MSs and all MSs have the same set of parameters, i.e., $\zeta_k=\zeta$, $\delta_k^2=\delta^2$, $\sigma_k^2=\sigma^2$, $e_k=e$, and $\gamma_k=\gamma$, $\forall k$. Moreover, we set $\zeta=0.5$, $\sigma^2=-70$dBm, and $\delta^2=-50$dBm in all simulations. It is further assumed that the signal attenuation from BS to all MSs is $40$dB corresponding to an identical distance of $5$ meters. With this transmission distance, the line-of-sight (LOS) signal is dominant, and thus the Rician fading is used
to model the channel. Specifically, $\bh_k$ is expressed as
\begin{equation}
\bh_k = \sqrt{\frac{K_R}{1+K_R}}\bh_k^{{\rm LOS}}+\sqrt{\frac{1}{1+K_R}}\bh_k^{{\rm NLOS}},
\end{equation}where $\bh_k^{{\rm LOS}}\in \mathbb{C}^{N_t\times 1}$ is the LOS deterministic component, $\bh_k^{{\rm NLOS}}\in \mathbb{C}^{N_t\times 1}$ denotes the Rayleigh fading
component with each element being a CSCG random variable with zero mean and covariance of $-40$dB, and $K_R$ is the Rician factor set to be $5$dB. Note that for the LOS component, we use the far-field uniform linear antenna array model \cite{Luo07} with $\bh_k^{{\rm LOS}}=10^{-4}[1~e^{j\theta_k}~e^{j2\theta_k}~\ldots~e^{j(N_t-1)\theta_k}~]^T$ with $\theta_k=-\frac{2\pi d \sin(\phi_k)}{\lambda}$, where $d$ is the spacing between successive antenna elements at BS, $\lambda$ is the carrier wavelength, and $\phi_k$ is the direction of ${\rm MS}_k$ to BS. We set $d=\frac{\lambda}{2}$, and $\{\phi_1,\phi_2,\phi_3,\phi_4\}=\{-30^{\textrm{o}}, -60^{\textrm{o}}, 60^{\textrm{o}}, 30^{\textrm{o}}\}$.

\begin{figure}
\centering
\includegraphics[width=3.5in]{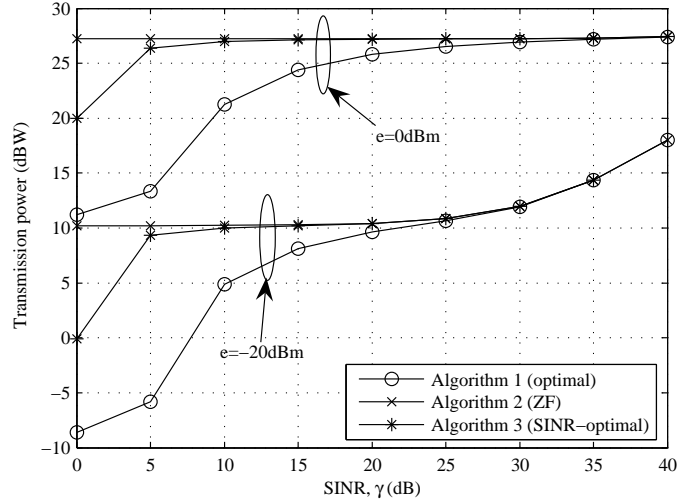}
\caption{Transmission power versus $\gamma$.}
\label{fig:fig2}
\end{figure}

First, we investigate the minimum transmission power required at BS versus the SINR target for all MSs, $\gamma$, with their harvested power constraint, $e$, being fixed. It is assumed that BS is equipped with $N_t=4$ transmit antennas. Fig. \ref{fig:fig2} shows the performance comparison by the optimal JBPS solution to problem \eqref{eq:P1} and the two suboptimal solutions based on ZF and SINR-optimal beamforming with $e=0$dBm or $e=-20$dBm. It is observed that as the harvested power constraint $e$ is increased from $-20$dBm to $0$dBm, substantially more transmission power is needed at BS for any each value of $\gamma$. It is also observed that for both the cases of $e=0$dBm and $e=-20$dBm, the minimum transmission power is achieved by the optimal JBPS solution for all values of $\gamma$. Moreover, when the SINR constraint $\gamma$ is small, the SINR-optimal based suboptimal solution is observed to achieve notably smaller transmission power than ZF based suboptimal solution. However, as $\gamma$ increases, the performance gap between the two suboptimal solutions vanishes. For example, when $\gamma>35$dB for the case of $e=0$dBm or $\gamma>25$dB for the case of $e=-20$dBm, the minimum transmission power values achieved by the two suboptimal solutions both converge to that by the optimal solution. Finally, it is observed that the transmission power with the ZF based suboptimal solution is not sensitive to the value of $\gamma$ when $\gamma$ is sufficiently small. The reason is as follows. In our simulation setup, the antenna noise power $\sigma^2$ is much smaller than the ID processing noise $\delta^2$. Thus, from \eqref{eq:r} in Remark \ref{remark2}, we have $\rho_k\approx\frac{\gamma_k\delta_k^2}{\frac{e_k}{\zeta_k}+\gamma_k\delta_k^2}$. In addition, if $\gamma$ is sufficiently small such that $\frac{e_k}{\zeta_k}$ is much larger than $\gamma_k\delta_k^2$, then we have $\rho_k\approx 0$. In this case, to meet the harvested power constrains in \eqref{eq:P1ZF}, it can be shown that $||\bar{\bv}_k^\ast||^2\approx\frac{e_k}{\zeta_k||\bU_k\bU_k^H\bh_k||^2}$, $\forall k$, which is independent of $\gamma$; thus, the total transmission power $\sum\limits_{k=1}^K||\bar{\bv}_k^\ast||^2$ is invariant over $\gamma$ in the small-$\gamma$ regime.
\begin{figure}
\centering
\includegraphics[width=3.5in]{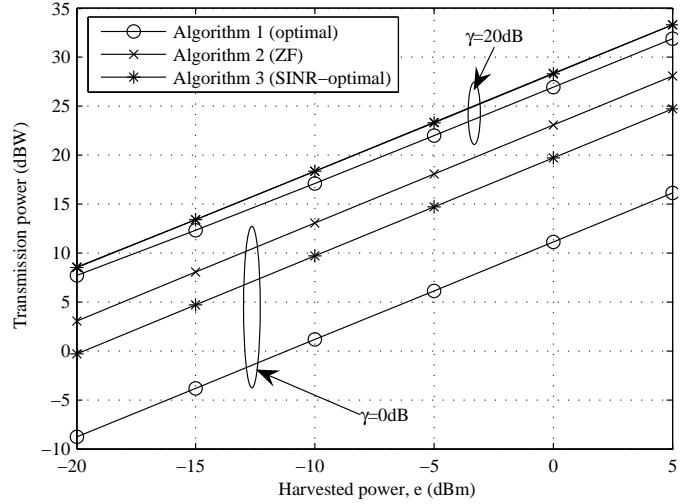}
\caption{Transmission power versus $e$.}
\label{fig:fig3}
\end{figure}

Next, we show in Fig. \ref{fig:fig3} the minimum transmission power achieved by the optimal and suboptimal algorithms over $e$ with fixed $\gamma=0$dB or $\gamma=20$dB. Similarly as in Fig. \ref{fig:fig2}, it is observed that for both the cases of $\gamma=0$dB and $\gamma=20$dB, the optimal solution achieves the minimum transmission power for all values of $e$. Furthermore, at low SINR, i.e., $\gamma=0$dB, the transmission power achieved by the SINR-optimal based solution is notably smaller than that by the ZF based solution, but all much larger than that by the optimal solution, for both values of $e$. However, at high SINR, i.e., $\gamma\geq20$dB, all the optimal and suboptimal solutions perform very closely to each other. \textcolor{blue}{This confirms our result in Proposition 5.3 regarding the asymptotic optimality of the suboptimal solutions.}

\begin{figure}
\centering
\includegraphics[width=3.5in]{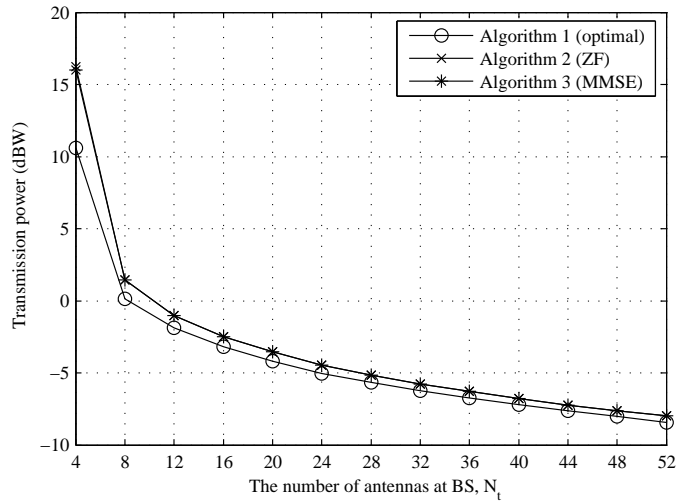}
\caption{Transmission power versus $N_t$ with fixed $e=-10$dBm and $\gamma=10$dB.}
\label{fig:fig4}
\end{figure}

At last, we investigate the impact of the number of transmit antennas, $N_t$, on the minimum transmission power for all proposed solutions, as shown in Fig. \ref{fig:fig4}, with fixed $\gamma=10$dB and $e=-10$dBm. It is observed that when the number of transmit antennas increases, the BS transmission power is substantially decreased for all solutions. This demonstrates the significant benefit by applying large or even massive antenna arrays for efficiently implementing MISO SWIPT systems in practice.

\section{Conclusion}
This paper has studied the joint transmit beamforming and receive power splitting design for a multiuser MISO
broadcast system for simultaneous wireless information and power transfer (SWIPT). The total transmission power at the BS is minimized subject to given SINR and harvested power constraints for MSs. The sufficient and necessary condition to guarantee the
feasibility of our problem is first derived. Then we solve this non-convex problem by applying the technique of SDR, and prove its optimality. Two suboptimal
designs of lower complexity than the optimal solution are also presented based on ZF and SINR-optimal beamforming, respectively,
and their performances are compared against the optimal
solution by simulations.

\appendices
\section{Proof of Lemma \ref{lem:no EH}}\label{appendix1}
First, it can be easily verified that if problem \eqref{eq:no EH} is not feasible, then problem \eqref{eq:P1} cannot be feasible since problem \eqref{eq:P1} has additional constraints on harvested power. Second, suppose problem \eqref{eq:no EH} is feasible, and let $\{\bv_k\}$ and $\{\rho_k\}$ be a feasible solution. It can be shown that the new solution $\{\alpha \bv_k\}$ and $\{\rho_k\}$ is also feasible to problem \eqref{eq:no EH}, $\forall \alpha>1$. Since there must exist a sufficiently large $\alpha>1$ such that the solution $\{\alpha \bv_k\}$ and $\{\rho_k\}$ also satisfies all the harvested power constraints of problem \eqref{eq:P1}, we can conclude that problem \eqref{eq:P1} is feasible. Lemma \ref{lem:no EH} is thus proved.

\section{Proof of Lemma \ref{lem:no rho}}\label{appendix2}
First, suppose problem \eqref{eq:no rho} is feasible and let $\{\bv_k\}$ denote a feasible solution. Then given any $0<\rho<1$, consider the following solution to problem \eqref{eq:no EH}: $\bar{\bv}_k=\bv_k/\sqrt{\rho}$, $\bar{\rho}_k=\rho$, $k=1,\cdots,K$. Since
\begin{align}
&\frac{\bar{\rho}_k|\bh_k^H\bar{\bv}_k|^2}{\sum_{j\neq k}\bar{\rho}_k|\bm{h}_k^H\bar{\bv}_j|^2+\bar{\rho}_k\sigma_k^2+\delta_k^2}\nonumber \\ =&\frac{|\bh_k^H\bv_k|^2}{\sum_{j\neq k}|\bm{h}_k^H\bv_j|^2+\rho \sigma_k^2+\delta_k^2}\nonumber \\ > &\frac{|\bh_k^H\bv_k|^2}{\sum_{j\neq k}|\bm{h}_k^H\bv_j|^2+ \sigma_k^2+\delta_k^2}\geq \gamma_k, ~~ \forall k,
\end{align}$\{\bar{\bv}_k\}$ and $\{\bar{\rho}_k\}$ is a feasible solution to problem \eqref{eq:no EH}. Therefore, if problem \eqref{eq:no rho} is feasible, then problem \eqref{eq:no EH} must be feasible.

Second, consider the case that problem \eqref{eq:no rho} is not feasible. In the following, we will show that problem \eqref{eq:no EH} cannot be feasible by contradiction. Suppose problem \eqref{eq:no EH} is feasible and let $\{\bv_k\}$ and $\{\rho_k\}$ be a feasible solution. Since $\rho_k< 1$, $\forall k$, we have
\begin{align}
\gamma_k \leq &\frac{\rho_k|\bh_k^H\bv_k|^2}{\sum_{j\neq k}\rho_k|\bm{h}_k^H\bv_j|^2+\rho_k\sigma_k^2+\delta_k^2}\nonumber \\ = & \frac{|\bh_k^H\bv_k|^2}{\sum_{j\neq k}|\bm{h}_k^H\bv_j|^2+\sigma_k^2+\frac{\delta_k^2}{\rho_k}} \nonumber \\ < & \frac{|\bh_k^H\bv_k|^2}{\sum_{j\neq k}|\bm{h}_k^H\bv_j|^2+\sigma_k^2+\delta_k^2}, ~~ \forall k.
\end{align}
As a result, $\{\bv_k\}$ is also a feasible solution to problem \eqref{eq:no rho}, which contradicts to the assumption that problem \eqref{eq:no rho} is not feasible.

To summarize, Lemma \ref{lem:no rho} is thus proved.

%\section{Proof of Proposition \ref{pro:condition}}\label{appendix3}
%Define the minimum mean square error of ${\rm UE}_k$ as $\varepsilon_k$. According to \cite[Theorem III.1]{Hunger10}, $\varepsilon_k$s are feasible if and only if $\sum_{k=1}^K\varepsilon_k \geq K-\rank{\bH}$. Since it is well-known that $\varepsilon_k=\frac{1}{1+\gamma_k}$, $\gamma_k$s are feasible if and only if $\sum_{k=1}^K\frac{1}{1+\gamma_k} \geq K-\rank(\bH)$. Proposition \ref{pro:condition} is thus proved.

\section{Proof of Proposition \ref{pro:rank one}}\label{appendix4}
First, we show the first part of Proposition \ref{pro:rank one} as follows. Since problem \eqref{eq:P3} is convex and satisfies the Slater's condition, its duality gap is zero \cite{cvx_book}. Let $\{\lambda_k\}$ and $\{\mu_k\}$ denote the dual variables associated with the SINR constraints and harvested power constraints of problem \eqref{eq:P3}, respectively. Then \textcolor{blue}{the partial Lagrangian of problem \eqref{eq:P3} is defined as
\begin{align}%\label{eq:Lagrangian}
L(\{\bX_k,\rho_k,\lambda_k,\mu_k\})&\triangleq\sum_{k=1}^K \trace(\bX_k)-\sum_{k=1}^K \lambda_k\left(\frac{1}{\gamma_k}\bh_k^H\bX_k\bh_k-\sum_{j\neq k}\bh_k^H\bX_j\bh_k-\sigma_k^2-\frac{\delta_k^2}{\rho_k}\right)\nonumber\\
&-\sum_{k=1}^K \mu_k\left(\sum_{j=1}^K \bh_k^H\bX_j\bh_k-\frac{e_k}{\zeta_k(1-\rho_k)}-\sigma_k^2\right).\nonumber
%\trace(\bA_k\bX_k)+\sum_{k=1}^K(\lambda_k-\mu_k)\sigma_k^2\nonumber \\ %&+\sum_{k=1}^K\left(\frac{\lambda_k\delta_k^2}{\rho_k}+\frac{\mu_k e_k}{\zeta_k(1-\rho_k)}\right),
\end{align}
With the Lagrangian function, the dual function of problem \eqref{eq:P3} is given by \cite[Sec. 5.7.3]{cvx_book}
\begin{align}%\label{eq:Lagrangian}
&\min_{\bX_k\succeq 0, 0<\rho_k<1, \forall k} L(\{\bX_k,\rho_k,\lambda_k,\mu_k\})\nonumber
\end{align}
which can be explicitly written as
\begin{align}\label{eq:Lagrangian}
&\min_{\bX_k\succeq 0, 0<\rho_k<1, \forall k} \left\{\sum_{k=1}^K\trace(\bA_k\bX_k)+\sum_{k=1}^K(\lambda_k-\mu_k)\sigma_k^2\right.\nonumber \\ &\left.+\sum_{k=1}^K\left(\frac{\lambda_k\delta_k^2}{\rho_k}+\frac{\mu_k e_k}{\zeta_k(1-\rho_k)}\right)\right\}
\end{align}}
where
\begin{equation}\label{eq:A}
\bA_k=\bI_{N_t}+\sum_{j=1}^K(\lambda_j-\mu_j)\bh_j\bh_j^H-\left(\frac{\lambda_k}{\gamma_k}+\lambda_k\right)\bh_k\bh_k^H, ~~ \forall k.
\end{equation}

Let $\{\lambda_k^\ast\}$ and $\{\mu_k^\ast\}$ denote the optimal dual solution to problem \eqref{eq:P3}. \textcolor{blue}{Accordingly, we define \begin{align}\bA_k^\ast=\bI_{N_t}+\sum\limits_{j=1}^K(\lambda_j^\ast-\mu_j^\ast)\bh_j\bh_j^H-\left(\frac{\lambda_k^\ast}{\gamma_k}+\lambda_k^\ast\right)\bh_k\bh_k^H.\end{align}
Then it is observed from \eqref{eq:Lagrangian} that, for any given $k$, $\bX_k^*$ must be a solution to the following problem:
\begin{equation}\label{eq:prob25}
\min_{\bX_k\succeq 0}\trace(\bA_k^\ast \bX_k).
\end{equation}
Note that, to guarantee a bounded dual optimal value,  we must have $$\bA^\ast_k\succeq 0, ~k=1,2,\ldots, K.$$ As a result, the optimal value of problem \eqref{eq:prob25} is zero, i.e., $\trace(\bA_k^\ast\bX_k^\ast)=0$, $k=1,2,\ldots, K$, which together with $\bA_k^\ast\succeq 0$ and $\bX_k^\ast\succeq 0$, $k=1,2,\ldots, K$, implies that
\begin{align}\label{eqn:kkt condition}
\bA_k^\ast \bX_k^\ast=\mathbf{0}, ~~ k=1,\cdots,K.
\end{align}
Moreover,} it is observed from \eqref{eq:Lagrangian} that
the optimal PS solution $\rho_k^\ast$ for any given $k\in \{1,\cdots,K\}$ must be a solution of the following problem:
\begin{equation}
\label{eq:rho}
\begin{split}
\textcolor{blue}{\min_{\rho_k}}  & ~~ \frac{\lambda_k^\ast\delta_k^2}{\rho_k}+\frac{\mu_k^\ast e_k}{\zeta_k(1-\rho_k)} \\
\st & ~~ 0< \rho_k< 1.
\end{split}
\end{equation}It is observed from the above problem that if $\lambda_k^\ast=0$ and $\mu_k^\ast>0$, the optimal solution will be $\rho_k^\ast \rightarrow 0$. Similarly, if $\mu_k^\ast=0$ and $\lambda_k^\ast>0$, then the optimal solution is $\rho_k^\ast \rightarrow 1$. Since given $e_k>0$ and $\gamma_k>0$, $\forall k$, $0<\rho_k^\ast < 1$ must hold for all $k$'s in problem \eqref{eq:P3}, the above two cases cannot happen.

Next, we show that $\lambda_k^\ast=0$ and $\mu_k^\ast=0$ cannot be true for any $k$ by contradiction. Suppose there exist some $k$'s such that $\lambda_k^\ast=\mu_k^\ast=0$. Define a set
\begin{equation}
\Psi\triangleq\{k|\lambda_k^\ast=0,\mu_k^\ast=0,1\leq k\leq K\},
\end{equation}
where $\Psi\neq \emptyset$. %From the Karush-Kuhn-Tucker (KKT) condition of problem \eqref{eq:P3} \cite{cvx_book}, it follows
%\begin{align}\label{eqn:kkt condition}
%\bA_k^\ast \bX_k^\ast=\mathbf{0}, ~~ k=1,\cdots,K,
%\end{align}where\begin{align}\bA_k^\ast=\bI_{N_t}+\sum\limits_{j=1}^K(\lambda_j^\ast-\mu_j^\ast)\bh_j\bh_j^H-\left(\frac{\lambda_k^\ast}{\gamma_k}+\lambda_k^\ast\right)\bh_k\bh_k^H.\end{align}
Define $\bB^\ast=\bI_{N_t}+\sum\limits_{j\notin \Psi}(\lambda_j^\ast-\mu_j^\ast)\bh_j\bh_j^H$.
Then $\bA_k^\ast$ can be expressed as
\begin{align}\label{eqn:impo}
\bA_k^\ast=\left\{\begin{array}{ll}\bB^\ast, & {\rm if} ~ k\in \Psi, \\ \bB^\ast-\left(\frac{\lambda_k^\ast}{\gamma_k}+\lambda_k^\ast\right)\bh_k\bh_k^H, & {\rm otherwise}.\end{array}\right.
\end{align}
%Furthermore, to guarantee that the Lagrangian of problem \eqref{eq:P3} is bounded from below such that the dual function exists, we have
%\begin{align}
%\bA_k^\ast\succeq \mathbf{0}, ~~ k=1,\cdots,K.
%\end{align}
Since $\bA_k^\ast\succeq \mathbf{0}$ and $-\left(\frac{\lambda_k^\ast}{\gamma_k}+\lambda_k^\ast\right)\bh_k\bh_k^H\preceq \mathbf{0}$, it follows that $\bB^\ast \succeq \mathbf{0}$. In the following, we show that $\bB^\ast\succ \mathbf{0}$ by contradiction. Suppose the minimum eigenvalue of $\bB^\ast$ is zero. Then, there exists at least an $\bx\neq \mathbf{0}$ such that $\bx^H\bB^\ast\bx=0$. According to (\ref{eqn:impo}), it follows that
\begin{align}\label{eq:eq30}
\bx^H\bA_k^\ast\bx=-\left(\frac{\lambda_k^\ast}{\gamma_k}+\lambda_k^\ast\right)\bx^H\bh_k\bh_k^H \bx\geq 0, ~~ k\notin \Psi.
\end{align}
Note that we have $\lambda^\ast_k>0$ if $k\notin \Psi$. \textcolor{blue}{Hence, we obtain from \eqref{eq:eq30} that $\vert\bh_k^H \bx\vert^2\leq 0$, $k\notin \Psi$}. It thus follows that
\begin{align}
\bh_k^H\bx=0, ~~ k\notin \Psi.
\end{align}Thus, we have
\begin{align}
\bx^H\bB^\ast\bx=\bx^H\left(\bI_{N_t}+\sum\limits_{j\notin \Psi}(\lambda_j^\ast-\mu_j^\ast)\bh_j\bh_j^H\right)\bx=\bx^H\bx>0,
\end{align}which contradicts to $\bx^H\bB^\ast\bx=0$. Thus, we have $\bB^\ast\succ \mathbf{0}$, i.e., $\rank(\bB^\ast)=N_t$. It thus follows from (\ref{eqn:impo}) that $\rank(\bA_k^\ast)=N_t$ if $k\in \Psi$. According to \eqref{eqn:kkt condition}, we have $\bX_k^\ast=\mathbf{0}$ if $k\in \Psi$. However, it is easily verified that $\bX_k^\ast=\mathbf{0}$ cannot be optimal for problem \eqref{eq:P3}. Therefore, it must follow that $\Psi=\emptyset$, i.e., $\lambda_k=0$ and $\mu_k=0$ cannot be true for any $k$. Since we have previously also shown that both the cases of $\lambda_k^\ast=0$, $\mu_k^\ast>0$ and $\lambda_k^\ast>0$, $\mu_k^\ast=0$ cannot be true for any $k$, it follows that $\lambda_k^\ast>0$, $\mu_k^\ast>0$, $\forall k$. According to the complementary slackness \cite{cvx_book}, we thus prove the first part of \textcolor{blue}{Proposition \ref{pro:rank one}}.

Next, we prove the second part of Proposition \ref{pro:rank one}. Since $\Psi=\emptyset$, it follows that $\bB^\ast=\bI_{N_t}+\sum\limits_{j=1}^K(\lambda_j^\ast-\mu_j^\ast)\bh_j\bh_j^H$, and \eqref{eqn:impo} thus reduces to \begin{align}
\bA_k^\ast=\bB^\ast-\left(\frac{\lambda_k^\ast}{\gamma_k}+\lambda_k^\ast\right)\bh_k\bh_k^H, ~~ k=1,\cdots,K.
\end{align}Since we have shown that $\rank(\bB^\ast)=N_t$, it follows that ${\rm rank}(\bA_k^\ast)\geq N_t-1$, $k=1,\cdots,K$. Note that if $\bA_k^\ast$ is of full rank, we have $\bX^\ast=\mathbf{0}$, which cannot be the optimal solution to problem \eqref{eq:P3}. As a result, it follows that ${\rm rank}(\bA_k^\ast)=N_t-1$, $\forall k$. According to \eqref{eqn:kkt condition}, we have ${\rm rank}(\bX_k^\ast)=1$, $k=1,\cdots,K$. The second part of Proposition \ref{pro:rank one} is thus proved. Combining the proofs of both parts, the proof of \textcolor{blue}{Proposition \ref{pro:rank one}} is thus completed.

\section{Proof of Proposition \ref{propZF}}
\label{appendix5}

Note that with ZF transmit beamforming, the SINR and harvested power constraints in problem \eqref{eq:P1ZF} can be decoupled over $k$. Moreover, it is observed that the objective function of problem \eqref{eq:P1ZF} is also separable over $k$. Hence, problem \eqref{eq:P1ZF} can be decomposed into $K$ subproblems, with the $k$-th subproblem, $k=1,\cdots,K$, expressed as
\begin{equation}
\label{eq:P1ZF2}
\begin{split}
\min_{\bv_k, \rho_k} & ~~ ||\bv_k||^2\\
\st & ~~ \frac{\rho_k|\bh_k^H\bv_k|^2}{\rho_k\sigma_k^2+\delta_k^2}{\geq} \gamma_k, \\
&~~\zeta_k(1-\rho_k)\left(|\bh_k^H\bv_k|^2+\sigma_k^2\right)\geq e_k, \\
&~~\bH_k^H\bv_k=0, \\
&~~0<\rho_k< 1.
\end{split}
\end{equation}

Next, we show that for problem \eqref{eq:P1ZF2}, with the optimal ZF beamforming solution $\bar{\bv}_k^\ast$ and PS solution $\bar{\rho}_k^\ast$, the SINR constraint and harvested power constraint should both hold with equality, by contradiction. First, suppose that both the two SINR and harvested power constraints are not tight given $\bar{\bv}_k^\ast$ and $\bar{\rho}_k^\ast$. In this case, there must exist an $\alpha_k$,  $0<\alpha_k<1$, such that with the new solution $\bv_k=\alpha_k\bar{\bv}_k^\ast$ and $\rho_k=\bar{\rho}_k^\ast$, either the SINR or harvested power constraint is tight. Moreover, with this new solution, the transmission power is reduced, which contradicts the fact that $\bar{\bv}_k^\ast$ and $\bar{\rho}_k^\ast$ is optimal for problem \eqref{eq:P1ZF2}. Thus, the case that both the SINR and harvested power constraints are not tight cannot be true. Next, consider the case when the SINR constraint is tight but the harvested power constraint is not tight. In this case, we can increase the value of $\bar{\rho}_k^\ast$ by a sufficiently small amount such that both the SINR and harvested power constraints become non-tight. Similar to the argument in the first case, we can conclude that this case cannot be true, too. Similarly, it can be shown that the case that the harvested power constraint is tight but the SINR constraint is not tight, also cannot be true. To summarize, with the optimal solution $\bar{\bv}_k^\ast$ and $\bar{\rho}_k^\ast$ for problem \eqref{eq:P1ZF2}, the SINR and harvested power constraints must both hold with equality. Hence, problem \eqref{eq:P1ZF2} is equivalent to
\begin{equation}
\label{eq:P1ZF2-eqv}
\begin{split}
\min_{\bv_k, \rho_k} & ~~ ||\bv_k||^2\\
\st & ~~ \frac{\rho_k|\bh_k^H\bv_k|^2}{\rho_k\sigma_k^2+\delta_k^2}= \gamma_k, \\
&~~\zeta_k(1-\rho_k)\left(|\bh_k^H\bv_k|^2+\sigma_k^2\right)= e_k, \\
&~~\bH_k^H\bv_k=0,\\
&~~0<\rho_k< 1.
\end{split}
\end{equation}

Note that in problem \eqref{eq:P1ZF2-eqv}, the first two equality constraints yield
\begin{equation}\label{eq:rho_equation}
\gamma_k\left(\sigma_k^2+\frac{\delta_k^2}{\rho_k}\right)=\frac{e_k}{\zeta_k(1-\rho_k)}-\sigma_k^2.
\end{equation}
We can rearrange \eqref{eq:rho_equation} as
$$\frac{\alpha_k}{1-\rho_k}-\frac{\beta_k}{\rho_k}=1,$$
which can be shown to have a unique solution satisfying $0<\rho_k<1$ given by
$$\bar{\rho}_k^*=\frac{-(\alpha_k+\beta_k-1)+\sqrt{(\alpha_k+\beta_k-1)^2+4\beta_k}}{2}.$$
Define $\bv_k=\sqrt{p_k}\barv_k$ with $||\barv_k||=1$, $\forall k$. Then problem \eqref{eq:P1ZF2-eqv} is equivalent to the following problem.
\begin{equation}
\label{eq:P1ZF2-eqv2}
\begin{split}
\min_{p_k,\barv_k} & ~~ p_k\\
\st & ~~ p_k|\bh_k^H\barv_k|^2= \tau_k, \\
&~~\bH_k^H\barv_k=0,\\
&~~||\barv_k||=1,
\end{split}
\end{equation}
where $\tau_k\triangleq\gamma_k\left(\sigma_k^2+\frac{\delta_k^2}{\bar{\rho}_k^*}\right)$.
It can be observed from the first constraint of problem \eqref{eq:P1ZF2-eqv2} that to achieve the minimum $p_k$, the optimal $\barv_k$ should be the optimal solution to the following problem:
\begin{equation}
\begin{split}
\max_{\barv_k} & ~~ |\bh_k^H\barv_k|^2\\
\st & ~~ \bH_k^H\barv_k=0, \\
& ~~ ||\barv_k||=1.
\end{split}
\end{equation}It can be shown that the unique (up to phase rotation) optimal solution to the above problem is given by
$$\barv_k=\frac{\bU_k\bU_k^H\bh_k}{||\bU_k\bU_k^H\bh_k||},$$where $\bU_k$ denotes the orthogonal basis for the null space of $\bH_k^H$.
Hence, the optimal power solution is given by
$$p_k=\frac{\tau_k}{|\bh_k^H\barv_k|^2}=\frac{\tau_k}{||\bU_k\bU_k^H\bh_k||^2}.$$
It thus follows that the optimal $\bar{\bv}_k^\ast$ for problem \eqref{eq:P1ZF2-eqv} is given by
$$\bar{\bv}_k^\ast=\sqrt{\gamma_k\left(\sigma_k^2+\frac{\delta_k^2}{\bar{\rho}_k^*}\right)}\frac{\bU_k\bU_k^H\bh_k}{||\bU_k\bU_k^H\bh_k||^2}.$$
Proposition \ref{propZF} is thus proved.

%\section{Proof of Proposition \ref{prop:asym_optimal}}
%\label{sec:ZF_optimality}
%Let $\bv_k=\sqrt{p_k}\bar{\bv}_k$ where $\bar{\bv}_k$ is the normalized beamformer with unit length. Due to the fact that $\sum_{j\neq k} \rho_kp_j|\bh_k^H\barv_j|^2\geq 0$, the SINR constraints of problem \eqref{eq:P1} implies that $\frac{\rho_kp_k|\bh_k^H\barv_k|^2}{\rho_k\sigma_k^2+\delta_k^2}\geq \gamma_k$, $\forall k$. Hence, when $\gamma_k\rightarrow\infty$, $\forall k\in\mK$, the transmission power for all users will go to infinity by noting that $|\bh_k^H\barv_k|^2$ is bounded. This implies that all the interference terms $|\bh_k^H\barv_j|^2$ ($\forall k\neq j$) must be zero when $\gamma_k\rightarrow\infty$, $\forall k\in\mK$ (otherwise, there exists at least one SINR which will approach to a finite number). Hence, the asymptotic optimality of zero-forcing follows.

\section{Proof of Proposition \ref{prop:subopt_alg2}}
\label{sec:subopt_alg2}
%\begin{proof}
First, we show the first part of Proposition \ref{prop:subopt_alg2} as follows.
Define
$$\phi_{I}(\rho_k, \alpha)\triangleq\frac{\rho_k\alpha|\bh_k^H\hat{\bv}_k|^2}{\sum_{j\neq k}\rho_k\alpha|\bm{h}_k^H\hat{\bv}_j|^2+\rho_k\sigma_k^2+\delta_k^2},$$ and
$$\phi_{E}(\rho_k, \alpha)\triangleq\zeta_k(1-\rho_k)(\alpha\sum_{j=1}^K |\bh_k^H\hat{\bv}_j|^2+\sigma_k^2).$$
Based on Lemmas \textcolor{blue}{3.1 and 3.2}, we can infer that problem \eqref{eq:power_min_SINR_only} is feasible if and only if problem \eqref{eq:P1} is feasible. Since $\{\hat{\bv}_k\}$ is the optimal solution to problem \eqref{eq:power_min_SINR_only}, it follows that,
\begin{equation}\label{eq:SINR_alpha}
\phi_I(1, 1)= \gamma_k,\forall k,
\end{equation}with $\rho_k=1$, $\forall k$, and $\alpha=1$. Since $\phi_{I}(\rho_k, \alpha)$ is a monotonically increasing function of $\alpha$ when $\alpha>0$, we can have an $\hat{\alpha}>1$ such that
\begin{equation}\label{eq:SINR_alpha1}
\phi_{I}(1, \hat{\alpha})> \gamma_k,\forall k.
\end{equation}
Since $\phi_{I}(\rho_k, \alpha)$ is a continuous function of $\rho_k$, we can find an $\hat{\rho}_k<1$ such that
\begin{equation}\label{eq:SINR_alpha2}
\phi_{I}(\hat{\rho}_k, \hat{\alpha})> \gamma_k,\forall k.
\end{equation}
Since $1-\hat{\rho}_k>0$, $\forall k$, $\phi_E(\hat{\rho}_k, \alpha)$ is increasing over $\alpha>0$. Together with the fact that $\phi_I(\hat{\rho}_k, \alpha)$ is also increasing over $\alpha>0$, $\forall k$, it follows that there must exist a sufficiently large $\hat{\alpha}$ such that
\begin{equation}
\phi_{I}(\hat{\rho}_k, \hat{\alpha})\geq\gamma_k,~\phi_{E}(\hat{\rho}_k, \hat{\alpha})\geq e_k,\forall k.
\end{equation}Therefore, we conclude that the ``if'' part is proved, i.e., problem \eqref{eq:P1-subopt2} is feasible if problem \eqref{eq:power_min_SINR_only} or \eqref{eq:P1} is feasible. The ``only if'' part can be shown easily since any feasible solution of problem \eqref{eq:P1-subopt2} must be feasible for problem \eqref{eq:power_min_SINR_only} or \eqref{eq:P1}. The first part of Proposition \ref{prop:subopt_alg2} is thus proved.

Next, we prove the second part of Proposition \ref{prop:subopt_alg2}.
With the defined $c_k$'s and $d_k$'s given in Proposition \ref{prop:subopt_alg2} and by noting that $\alpha c_k-\sigma_k^2\geq\delta_k^2>0$, $\forall k$, when $\alpha\geq1$, problem \eqref{eq:P1-subopt2} can be equivalently rewritten as
\begin{equation}
\label{eq:P1-subopt21}
\begin{split}
&\min_{\alpha, \{\rho_k\}}  ~~ \alpha\\
& \st ~~ \rho_k\geq \frac{\delta_k^2}{\alpha c_k-\sigma_k^2}, ~~ \forall k,\\
&~~1-\rho_k\geq \frac{e_k}{\zeta_k(\alpha d_k+\sigma_k^2)}, ~~ \forall k, \\
&~~0< \rho_k< 1, ~~ \forall k,\\
&~~\alpha>1.
\end{split}
\end{equation}
It can be shown that problem \eqref{eq:P1-subopt21} is equivalent to the following problem:
\begin{equation}
\label{eq:P1-subopt23}
\begin{split}
\min_{\alpha>1}  & ~~ \alpha\\
\st & ~~ g_k(\alpha)\leq 1, ~~ \forall k,
\end{split}
\end{equation}
where $g_k(\alpha)\triangleq\frac{\delta_k^2}{\alpha c_k-\sigma_k^2}+\frac{e_k}{\zeta_k(\alpha d_k+\sigma_k^2)}$, since the two problems have the same optimal value.

It is observed that $g_k(\alpha)=1$ is a quadratic equation. Let $\underline{\alpha}_k$ and $\bar{\alpha}_k$, where $\bar{\alpha}_k\geq\underline{\alpha}_k$, be the two roots of the equation $g_k(\alpha)=1$.
By noting that $\frac{\delta_k^2}{c_k-\sigma_k^2}=1$ due to \eqref{eq:SINR_alpha}, we conclude that $g_k(1)>1$, implying $\bar{\alpha}_k>\underline{\alpha}_k$ (otherwise any $\alpha$ satisfies $g_k(\alpha)\leq 1$). On the other hand, due to the fact that $c_kd_k\zeta_k>0$, the inequality $g_k(\alpha)\leq 1$ implies either $\alpha\geq \bar{\alpha}_k$ or $\alpha\leq\underline{\alpha}_k$. Since $g_k(\alpha)$ is a monotonically decreasing function of $\alpha$ for $\alpha> 1$ and $g_k(1)>1$, we must have $\underline{\alpha}_k<1<\bar{\alpha}_k$. It thus follows that problem \eqref{eq:P1-subopt22} is simplified as
\begin{equation}
\label{eq:P1-subopt22}
\begin{split}
\min_{\alpha}  & ~~ \alpha\\
\st & ~~ \alpha\geq \bar{\alpha}_k, ~~ \forall k,
\end{split}
\end{equation}
which has the optimal solution given by $\tilde{\alpha}^*=\max_{1\leq k\leq K} \bar{\alpha}_k$. Furthermore, we can easily check that
$\tilde{\rho}_k^*=\frac{\delta_k^2}{\tilde{\alpha}^* c_k-\sigma_k^2}$, $\forall k$, is the corresponding optimal solution to problem \eqref{eq:P1-subopt21} with given $\tilde{\alpha}^\ast$. This thus completes the proof of the second part of Proposition \ref{prop:subopt_alg2}. By combining the proofs of both parts, the proof of Proposition \ref{prop:subopt_alg2} is thus completed.
%\end{proof}

\textcolor{blue}{\section{Proof of Proposition \ref{prop:asym_opt}}}
\label{sec:proof_asym_opt}
For each $k$, let $\bv_k=\sqrt{p_k}\bar{\bv}_k$ where $p_k$ denotes the power allocation and $\bar{\bv}_k$ denotes the normalized beamformer with unit norm, i.e., the beam direction. Then the SINR constraints can be expressed as
\begin{equation}\label{prop:SINR_expression}
\frac{\rho_kp_k\vert\bh_k^H\bar{\bv}_k\vert^2}{\sum_{j\neq k} \rho_kp_j\vert\bh_k^H\bar{\bv}_j\vert^2+\rho_k\sigma_k^2+\delta_k^2}\geq \gamma_k, \forall k.
\end{equation}
We first prove a basic result that, for any set of beamforming vectors $\{\bv_k\}$ that satisfies the SINR constraints with $\gamma_k\rightarrow\infty$, $\forall k$,  $p_k$'s must go to infinity and furthermore $\bar{\bv}_k$'s must satisfy zero-forcing condition, i.e, $\vert\bh_k^H\bar{\bv}_j\vert\rightarrow 0$, $\forall j\neq k$.

Due to the fact that $\sum_{j\neq k} \rho_kp_j\vert\bh_k^H\bar{\bv}_j\vert^2\geq 0$, $\forall k$, the SINR constraints imply
$$\frac{\rho_kp_k\vert\bh_k^H\bar{\bv}_k\vert^2}{\rho_k\sigma_k^2+\delta_k^2}\geq \gamma_k, \forall k.$$ By noting that $\vert\bh_k\bar{\bv}_k\vert^2$ is bounded, we thus infer from the above inequalities that,  for any $k$, $p_k\rightarrow \infty$ when $\gamma_k\rightarrow \infty$. Next we show by contradiction that zero-forcing condition must be satisfied for $\{\bv_{k}\}$ when $\gamma_k\rightarrow\infty$, $\forall k$. Suppose we have $\bh_k^H\bar{\bv}_j\neq 0$ for any $j\neq k$ when $\gamma_k\rightarrow\infty$, $\forall k$. Note that all the terms $\vert\bh_k^H\bv_j\vert^2$ ($\forall k,j$) are bounded.  The coupling of $p_k$'s in all the SINR constraints \eqref{prop:SINR_expression} implies that $p_k$'s need to be in the same order. Combining this with the fact that for all $k$, $p_k\rightarrow \infty$ as $\gamma_k\rightarrow\infty$, we can assume that without loss of optimality $p_k=a_kp$, $\forall k$, with $a_k$'s being constant and $p\rightarrow\infty$ (when $\gamma_k\rightarrow \infty$). It thus follows from \eqref{prop:SINR_expression} that
\begin{equation}
\frac{\rho_ka_k\vert\bh_k^H\bar{\bv}_k\vert^2}{\sum_{j\neq k} \rho_k a_j\vert\bh_k^H\bar{\bv}_j\vert^2}\geq \infty, \forall k,
\end{equation}
implying a contradiction. To summarize, for any $\bv_k$'s that satisfy the SINR constraints with $\gamma_k\rightarrow \infty$, $\forall k$, we have $p_k\rightarrow \infty$ and $\vert\bh_k^H\bar{\bv}_j\vert\rightarrow 0$, $\forall j\neq k$, $k=1,2,\ldots, K$.

From the above result, we conclude that, for all the three solutions provided by Algorithms 1, 2, and 3, the power allocation will become infinity as $\gamma_k$'s increase to infinity. Moreover, the optimal solution and the SINR-optimal beamforming based suboptimal solution must satisfy zero-forcing condition (as the ZF-based suboptimal solution does) when $\gamma_k$'s go to infinity. Hence, the three algorithms will asymptotically achieve the same minimum transmit power when $\gamma_k\rightarrow \infty$, $\forall k$. Thus, the proof of Proposition \ref{prop:asym_opt} is completed.


\begin{thebibliography}{1}

\bibitem{ZhangArXiv}
R. Zhang and C. Ho, ``MIMO broadcasting for simultaneous
wireless information and power transfer,'' \emph{IEEE Trans. Wireless Commun.}, vol. 12, no. 5, pp. 1989-2001, May 2013.

\bibitem{Rui12} X. Zhou, R. Zhang, and C. Ho, ``Wireless information and power transfer: architecture design and rate-energy tradeoff,'' {\it IEEE Trans. Commun.}, vol. 61, no. 11, pp. 4754-4767, Nov. 2013.

\bibitem{Liu2013}
L. Liu, R. Zhang, and K. C. Chua, ``Wireless information transfer with opportunistic energy harvesting,'' \emph{IEEE Trans. Wireless Commun.}, vol. 12, no. 1, pp. 288-300, Jan. 2013.

\bibitem{Rui13TCOM} L. Liu, R. Zhang, and K. C. Chua, ``Wireless information and power transfer: a dynamic power splitting approach,'' {\it IEEE Trans. Commun.}, vol. 61, no. 9, pp. 3990-4001, Sept. 2013.



\bibitem{Xiang2012}
Z. Xiang and M. Tao, ``Robust beamforming for wireless information
and power transmission,'' \emph{IEEE Wireless Commun. Letters},
vol. 1, no. 4, pp. 372-375, 2012.

\bibitem{RuiWCNC} H. Ju and R. Zhang, ``A novel mode switching scheme utilizing random beamforming for opportunistic energy harvesting,'' in \emph{Proc. IEEE Wireless Communications and Networking Conference (WCNC)}, pp. 4250-4255, Apr. 2013.

\bibitem{Rui13} J. Xu, L. Liu, and R. Zhang, ``Multiuser MISO beamforming for simultaneous wireless information and power transfer,'' in {\it Proc. IEEE International Conference on Acoustics, Speech, and Signal Processing (ICASSP)}, 2013.

\bibitem{RuiGlobecom} L. Liu, R. Zhang, and K. C. Chua, ``Secrecy wireless information and power transfer with MISO beamforming,'' in {\it Proc. IEEE Global Communications Conference (Globecom)}, 2013.

\bibitem{Clerckx13} J. Park and B. Clerckx, ``Joint wireless information and energy transfer in a two-user MIMO interference channel,'' \emph{IEEE Trans. Wireless Commun.}, vol. 12, no. 8, pp. 4210-4221, Aug. 2013.

\bibitem{Ottersten2013}
S. Timotheou, I. Krikidis, and B. Ottersten, ``MISO interference channel with QoS and RF energy harvesting constraints,'' in \emph{Proc. IEEE ICC}, 2013.

\bibitem{Shen12} C. Shen, W. C. Li, and T. H. Chang, ``Simultaneous information and energy transfer: a two-user MISO interference channel case,'' in {\it Proc. IEEE Global Commun. Conf. (GLOBECOM)}, Dec. 2012.


\bibitem{Wu} Y. Wu, Y. Liu, Q. Xue, S. Li, and C. Yu, ``Analytical design method of
multiway dual-band planar power dividers with arbitrary power division,'' {\it IEEE Trans. Microwave Theory and Techniques}, vol. 58, no. 12, pp. 3832-3841, Dec. 2010.


%\bibitem{Chalise2012}
%B. K. Chalise, Y. D. Zhang, and M. G. Amin, ``Energy harvesting
%in an OSTBC based amplify-and-forward MIMO relay system,''
%in \emph{Proc. IEEE ICASSP}, pp. 3201-3204, Mar. 2012.

%\bibitem{Fouladgar2012}
%A. M. Fouladgar and O. Simeone, "On the transfer of information and energy in multi-user systems," \emph{IEEE Commun. Letters}, vol. 16, no. 11, pp. 1733-1736, Nov. 2012.

\bibitem{SunH2012}
H. Sun, Y. -X. Guo, M. He, and Z. Zhong, ``Design of a high-efficiency 2.45-GHz rectenna
for low-input-power energy harvesting,'' \emph{IEEE Ant. Wireless Propag. Lett.}, vol. 11, pp.
929-932, 2012.



\bibitem{Luo2010}
Z.-Q. Luo, W.-K. Ma, A.M.-C. So, Y. Ye, and S. Zhang,
``Semidefinite relaxation of quadratic optimization problems,''
\emph{IEEE Trans. Signal Process. Mag.}, vol. 27, no. 3, pp. 20-34, 2010.

%\bibitem{sedumi}
%J.~F.~Sturm, ``Using SeDuMi 1.02, a Matlab toolbox for optimization
%over symmetric cones'', \emph{Optimiz. Meth. Softw.}, vol.~11--12,
%pp.~625--653, 1999.





\bibitem{Ottersten2001}
M. Bengtasson and B. Ottersten, ``Optimal and suboptimal transmit beamforming,''
Chapter 18 in \emph{Handbook of Antennas in Wireless Communications}, L. C. Godara, Ed., CRC Press, Aug. 2001.

\bibitem{Codreanu07} M. Codreanu, A. Tolli, M. Juntti, and M.
Latva-aho, ``Joint design of Tx-Rx beamformers in MIMO downlink
channels,'' {\it IEEE Trans. Signal Process.}, vol. 55, no. 9, pp.
4639-4655, Sep. 2007.

\bibitem{Boche2004}
M. Schubert and H. Boche, ``Solution of the multiuser downlink beamforming
problem with individual SINR constraints,'' {\it IEEE Trans. Vehic. Tech.}, vol. 53, no. 1, pp. 18-28, Jan. 2004.



\bibitem{Hunger10} R. Hunger and M. Joham, ``A complete description of the QoS feasibility region in the vector broadcast channel,'' \emph{IEEE Trans. Signal Process.}, vol. 58, no. 7, pp. 3870-3878, July 2010.

\bibitem{cvx_book}
S. Boyd and L. Vandenberghe, \emph{Convex Optimization}. Cambridge U.K.:
Cambridge Univ. Press, 2004.



%\bibitem{Bertsekas_book}
%D. Bertsekas, \emph{Nonlinear Programming}, 2nd ed. Belmont, MA: Athena Scientific, 1999.

%\bibitem{Horn_book}
%R. A. Horn and C. R. Johnson, \emph{Topics in Matrix Analysis}. Cambridge,
%U.K.: Cambridge Univ. Press, 1994.


%\bibitem{Eldar2008}
%A. Wiesel, Y. C. Eldar, S. Shamai, ``Zero-forcing precoding and generalized inverses,'' \emph{IEEE Trans. Signal Process.}, vol. 56, no. 9, pp. 4409-4418, Sept. 2008.





\bibitem{cvx2012}
M. Grant and S. Boyd. CVX: Matlab software for disciplined convex programming, version 2.0 beta, Sept. 2012 [online]. Available on-line at http://cvxr.com/cvx.

\bibitem{Luo07} E. Karipidis, N. D. Sidiropoulos, and Z. Q. Luo, ``Far-field multicast beamforming for uniform linear antenna arrays,'' {\it IEEE Trans. Signal Process.,} vol. 55, no. 10, pp. 4916-4927, Oct. 2007.
\bibitem{Ben2001}
B.-T. Aharon and A. Nemirovski, \emph{Lectures on Modern Convex Optimization: Analysis, Algorithms, and Engineering Applications}, MOS-SIAM Series on Optimization, 2001.

%\bibitem{Yoo2006}
%T. Yoo and A. Goldsmith, ``On the optimality of multiantenna
%broadcast scheduling using zero-forcing beamforming,'' \emph{IEEE
%Journal on Selected Areas in Communications}, vol. 24, no. 3, pp.
%528-541, 2006.

\end{thebibliography}
\end{document}